\documentclass[envcountsame]{llncs}
    \usepackage{amsmath}
    \usepackage{amsxtra}
    \usepackage{amstext}
    \usepackage{amssymb}
    \usepackage{latexsym}
    \usepackage{dsfont} 
\newcommand{\GF}[1]{{\mathbb F}_{#1}}

\begin{document}
\title{Solving $X^{2^{2k}+2^{k}+1}+(X+1)^{2^{2k}+2^{k}+1}=b$ over $\GF{2^{4k}}$}
\author{  Kwang Ho Kim\inst{1,2} \and  Sihem Mesnager\inst{3} \and Chung Hyok Kim\inst{1}} \institute{ Institute of Mathematics,
State Academy of Sciences, Pyongyang, Democratic People's Republic
of Korea. \email{khk.cryptech@gmail.com}  \and Department of
Mathematics, University of Paris VIII, F-93526 Saint-Denis,
University Sorbonne Paris Cit\'e, LAGA, UMR 7539, CNRS, 93430
Villetaneuse and Telecom Paris, Polytechnic Institute of Paris,
91120 Palaiseau, France. \email{smesnager@univ-paris8.fr} \and
PGItech Corp., Pyongyang, Democratic People's Republic of Korea. }
\institute{Institute of Mathematics, State Academy of Sciences,
Pyongyang, Democratic People's Republic of Korea\\
\email{khk.cryptech@gmail.com} \and PGItech Corp., Pyongyang, Democratic People's Republic of Korea\\ \and Department of Mathematics, University of Paris VIII, F-93526 Saint-Denis, University Sorbonne Paris Cit\'e, LAGA, UMR 7539, CNRS, 93430 Villetaneuse and Telecom Paris, Polytechnic Institute of Paris, 91120 Palaiseau, France.\\
\email{smesnager@univ-paris8.fr}\\} \maketitle

\begin{abstract} 
 Let $F(X)=X^{2^{2k}+2^k+1}$ be the power function over the finite
field $\GF{2^{4k}}$ which is known as the Bracken-Leander function.
In \cite{BCC10,BL10,CV20,Fu22,XY17}, it was proved that the number
of solutions in $\GF{q^4}$ to the equation $F(X)+F(X+1)=b$ is in
$\{0,2,4\}$ for any $b\in \GF{q^4}$ and the number of the $b$ giving
$i$ solutions have been determined for every $i$. However, no paper provided a direct and complete method to solve such an equation, and 
  this problem remained open. This article presents a direct technique to derive
an explicit solution to that equation. The main result
 in \cite{BCC10,BL10,Fu22,XY17}, determining differential spectrum of
 $F(X)=X^{2^{2k}+2^k+1}$ over $\GF{2^{4k}}$,
 is re-derived simply from our results.
\end{abstract}


\section{Introduction}

Let $F(X) = X^d$ be a monomial over $\GF{2^n}$. Define
$$N_i := \# \{b \in \GF{2^n} \mid \#\{x \in \GF{2^n} \mid F(x)+F(x+1)=b\} = i\}.$$ The \textsl{differential spectrum} of $F(X)$  over $\GF{2^n}$ is defined as $\mathbb{S}
= \{N_i\}_{0\leq i}.$ The differential spectrum is an important
concept in cryptography as it quantifies the degree of security of
$F(X)$  with respect to differential attacks \cite{BS91} when it is
used as the \textsl{Substitution box} (S-box) in the cipher.  Hence
the problem of computing the differential spectrum was of great interest
for cryptographical practices as well as from a mathematical point of
view and deserved much attention in the literature.

When $d=2^{2k}+2^{k}+1$, $F(X)=X^d$ over $\GF{2^{4k}}$ is called the
Bracken-Leander function as  it was firstly proved in \cite{BL10} by
Bracken and Leander that the maximum among its differential
spectrum called its differential uniformity, is 4 when $k$ is odd.
Dobbertin (\cite{Dobbertin98}) first introduced the power function
in 1998 to meet the conjectured nonlinearity
bound and hence is also called a highly nonlinear function.
Blondeau, Canteaut and Charpin (\cite{BCC10}) have firstly  and
completely determined the differential spectrum of the
Bracken-Leander function for the case of odd $k$. Then Xiong and Yan
\cite{XY17} followed the method of Bracken and Leander (\cite{BL10})
Furthermore, it showed that it also holds for even $k$. The result in
\cite{BCC10} for odd $k$ was obtained indirectly by depending on the
non-trivial relation between power functions and the corresponding
cyclic codes and the argument was quite involved, as noted by the
authors  of \cite{XY17}. For the even case, the differential
the spectrum of $F(X)$ was obtained directly but with great skills by
analyzing the exact number of solutions of derived equations as
noted by the author of \cite{Fu22} who provided another shorter
proof to determine the differential spectrum of $F(X)$.
However, it remained open to directly solving that equation in the
definition of the differential spectrum. This article explicitly
solves the equation $X^{2^{2k}+2^{k}+1}+(X+1)^{2^{2k}+2^{k}+1}=b$
over $\GF{2^{4k}}$ for any positive integer $k$ (Theorem~\ref{main
theorem}). The differential spectrum of the Bracken-Leander function
is also determined as a consequence of our result
(Corollary~\ref{cor}).

\section{Prerequisites}
Throughout this paper, $k$ is any positive integer and  $q=2^k$.
 Given positive integers $l$ and $L$ with $l\mid L$, define a linearized polynomial
\[\mathbf{T}^{L}_{l}(X):=X+X^{2^l}+\cdots+X^{2^{l\left(\frac{L}{l}-2\right)}}+X^{2^{l\left(\frac{L}{l}-1\right)}}.\]
When $l=1$, we will abbreviate $\mathbf{T}^{L}_{1}(\cdot)$ as
$\mathbf{T}_L(\cdot)$. For $x\in \GF{2^L}$, $\mathbf{T}_l^L(x)$ is
the trace $\mathbf{Tr}_{l}^{L}(x)$ of $x$ over $\GF{2^l}$. These
polynomials have many interesting properties and were extensively
studied in \cite{MK20}. Among them, we recall a result that will be
very frequently used throughout this article.
\begin{proposition}[{\cite[Lemma 1, Proposition 1]{MK20}}]\label{P1}
Let $l$ and $L$ be positive integers such as $l\mid L$. Then,
for any $x\in \overline{\GF{2}}$, the followings are true.
\begin{enumerate}
\item\label{Tp1}
\[
\mathbf{T}_l(x+x^{2})=x+x^{2^l}.
\]
\item\label{Tp2}
\[
\mathbf{T}_L(x)=\mathbf{T}_l\left(\mathbf{T}_l^L(x)\right)=\mathbf{T}_l^L\left(\mathbf{T}_l(x)\right).
\]
\end{enumerate}\qed
\end{proposition}

Quadratic equations over finite fields were completely solved to
which we will reduce our study on the equation
$X^{2^{2k}+2^{k}+1}+(X+1)^{2^{2k}+2^{k}+1}=b$.
\begin{proposition}[{\cite{BSS1999,MP2013,MK20}}]\label{quadratic_old}
The quadratic equation
\[
x^2+x+a=0, a\in \GF{2^n}
\]
has solutions in $\GF{2^n}$ if{f}
\[
\mathbf{T}_n(a)=0.
\]
Let us assume $\mathbf{T}_n(a)=0$. Let $\delta$ be an element in
$\GF{2^n}$ such that $\mathbf{T}_n(\delta)=1$ (if $n$ is odd, then
one can take $\delta=1$). Then,
\[x_0=\sum_{i=0}^{n-2}(\sum_{j=i+1}^{n-1}\delta^{2^j})a^{2^i}\] is a
solution to the equation.\qed
\end{proposition}

\section{Main Results}
Let $f(X):=X^{q^2+q+1}+(X+1)^{q^2+q+1}$. The following fact is very
usefully exploited afterwards.
\begin{proposition}\label{alpha}
Let $\alpha\in \GF{q}$, $\beta\in \GF{q^2}$ and $x\in \GF{q^4}$.
Then,
\[
f(x+\alpha)=f(x)+\alpha+\alpha^2.
\]
\[
f(x+\beta)=f(x)+(\beta+\beta^q)(x+x^{q^2})+\beta^2+\beta^q.
\]
\end{proposition}
\begin{proof}
We can rewrite
\[
f(X)=X^{q^2+q}+X^{q^2+1}+X^{q+1}+X^{q^2}+X^q+X+1.
\]

Then,
\begin{align*}
&f(x+\alpha)=\left(x^{q^2+q}+\alpha(x^{q^2}+x^q)+\alpha^2\right)+\left(x^{q^2+1}+\alpha(x^{q^2}+x)+\alpha^2\right)+\\
&\quad\quad\left(x^{q+1}+\alpha(x^{q}+x)+\alpha^2\right)+(x^{q^2}+\alpha)
+(x^{q}+\alpha)+(x+\alpha)+1\\
 &=f(x)+\alpha+\alpha^2,\\
\end{align*}
and
\begin{align*}
&f(x+\beta)=\left(x^{q^2+q}+\beta^q x^{q^2}+\beta x^q+\beta^{q+1}\right)+\left(x^{q^2+1}+\beta(x^{q^2}+x)+\beta^2\right)+\\
&\quad\quad\left(x^{q+1}+\beta x^{q}+\beta^q
x+\beta^{q+1}\right)+(x^{q^2}+\beta)
+(x^{q}+\beta^q)+(x+\beta)+1\\
 &=f(x)+(\beta+\beta^q)(x+x^{q^2})+\beta^2+\beta^q.
\end{align*}\qed
\end{proof}
\subsection{Solving $f(X)=b$}
Let $x\in \GF{q^4}$ be a solution to $f(X)=b$, i.e.,
\begin{equation}\label{eq1}
x^{q^2+q}+x^{q^2+1}+x^{q+1}+x^{q^2}+x^q+x+1=b.
\end{equation}
By raising  the equation~\eqref{eq1} to the $q$-th power we get
\begin{equation}\label{eq2}
x^{q^3+q^2}+x^{q^3+q}+x^{q^2+q}+x^{q^3}+x^{q^2}+x^q+1=b^q.
\end{equation}
Summing the equations~\eqref{eq1} and \eqref{eq2} produces
\begin{align*}
&x^{q^3+q^2}+x^{q^3+q}+x^{q^2+1}+x^{q+1}+x^{q^3}+x=b+b^q\\
&\Longleftrightarrow(x+x^q)^{q^3+q}+(x+x^q)^{q^3}=b+b^q,
\end{align*}
or equivalently,
\begin{equation}\label{eq5}
(x+x^q)^{q^2+1}+(x+x^q)=b^q+b^{q^2}.
\end{equation} Let $e:=\mathbf{T}^{4k}_k(b)$. Then, $\eqref{eq5}+\eqref{eq5}^{q^2}$ gives
\begin{equation}\label{trcond}
\mathbf{T}^{4k}_k(x)=e.
\end{equation}
Therefore, we have $(x+x^q)^{q^2}=e+(x+x^q)$ and substituting it
into \eqref{eq5} results in
\begin{equation}\label{eq6}
(x+x^q)^2+\left(e+1\right)(x+x^q)+b^q+b^{q^2}=0.
\end{equation}
Now, the case $e=1$ is easily solved.
\begin{proposition}\label{prop1}
Assume $\mathbf{T}^{4k}_k(b)=1$. Then, the followings hold.
\begin{enumerate}
\item The equation~\eqref{eq1} has 0 or 2 solutions in $\GF{q^4}$.
\item The equation~\eqref{eq1} has  2 solutions in
$\GF{q^4}$ if{f}
$$\mathbf{T}_{2k}\left(b^{q^2+1}\right)=1.$$
\item Under the conditions $\mathbf{T}^{4k}_k(b)=1$ and
$\mathbf{T}_{2k}(b^{q^2+1})=1$, the two $\GF{q^4}-$solutions to
\eqref{eq1} are
\begin{equation}\label{sol1x}
x_{i}=b^{q/2}+\alpha_i, i\in
\{0,1\}
\end{equation}
where $\alpha_0$ and $\alpha_1$ are two
solutions in $\GF{q}$ to
\begin{equation}\label{sol1a}
\alpha^2+\alpha=f(b^{\frac{q}{2}})+b.
\end{equation}
\end{enumerate}
\end{proposition}
\begin{proof}
 If $\mathbf{T}^{4k}_k(b)=1$, then the equation~\eqref{eq6} transforms into
 $x+x^q=(b+b^q)^{q/2}$ which indicates
\[
x=b^{q/2}+\alpha \text{ for some } \alpha\in \mathbb{F}_q.
\]
Thanks to Proposition~\ref{alpha}, substituting $x=b^{q/2}+\alpha$
to \eqref{eq1} yields
 $\alpha^2+\alpha=\left(f(b^{q})+b^2\right)^{\frac{1}{2}}$. To complete the proof, it remains only to confirm that
\begin{align*}
&f(b^{q})+b^2\overset{\mathbf{T}^{4k}_k(b)=1}{=}b^{q^3+q^2}+b^{q^3+q}+b^{q^2+q}
+b^2+b\\
&=b^{q^2}(b^{q^3}+b^q)+b^{q^3+q}+b^2+b\\
&\overset{\mathbf{T}^{4k}_k(b)=1}{=} b^{q^2}(b^{q^2}+b+1)+b^{q^3+q}+b^2+b\\
&=\left(b^{q^2+1}+b^{q(q^2+1)}\right)+(b+b^{q^2})+(b+b^{q^2})^2\\
&=\mathbf{T}_k^{2k}\left(b^{q^2+1}\right)+\mathbf{T}_{2k}^{4k}(b+b^2),
\end{align*}
 and subsequently
\begin{align*}
&\mathbf{T}_k\left(f(b^{q})+b^2\right)\overset{\text{Prop.
\ref{P1}}}{=}\mathbf{T}_{2k}\left(b^{q^2+1}\right)+\mathbf{T}_{2k}^{4k}(b+b^q)\\
&\overset{\text{Prop.
\ref{P1}}}{=}\mathbf{T}_{2k}\left(b^{q^2+1}\right)+\mathbf{T}_{4k}(b)=\mathbf{T}_{2k}\left(b^{q^2+1}\right)+1,
\end{align*}
then we use Proposition~\ref{quadratic_old}. \qed
\end{proof}

Then, we should consider the case $e\neq 1.$ In this case, from the
equation~\eqref{eq6} we derive
$$x+x^q=(e+1)\left(\mathbf{T}_k\left(\frac{b^q}{e^2+1}\right)+\mathbb{F}_2\right),$$
which can be directly checked as
$\mathbf{T}_k\left(\frac{b^q}{e^2+1}\right)+\mathbf{T}_k\left(\frac{b^q}{e^2+1}\right)^2\overset{\text{Prop.
\ref{P1}}}{=}\frac{b^q}{e^2+1}+\left(\frac{b^q}{e^2+1}\right)^q=\frac{b^q+b^{q^2}}{e^2+1}.$
Notice via Proposition~\ref{quadratic_old} there exists an element
$c\in \mathbb{F}_{q^4}$ satisfying
$$c^2+(e+1)c=b^q,$$
or equivalently,
$$\left(\frac{c}{e+1}\right)^2+\frac{c}{e+1}=\frac{b^q}{e^2+1}$$
 since
$\mathbf{T}_{4k}\left(\frac{b^q}{e^2+1}\right)=\mathbf{T}_k\left(\frac{e}{e^2+1}\right)=\mathbf{T}_k\left(\frac{e}{e+1}+\left(\frac{e}{e+1}\right)^2\right)\overset{\text{Prop.
\ref{P1}}}{=}\frac{e}{e+1}+\frac{e}{e+1}=0$. Therefore, we have
\begin{align*}
x+x^q&=(e+1)\left(\mathbf{T}_k\left(\frac{c}{e+1}+\left(\frac{c}{e+1}\right)^2\right)+\mathbb{F}_2\right)
\\&\overset{\text{Prop. \ref{P1}}}{=}(e+1)\left(\frac{c}{e+1}+\left(\frac{c}{e+1}\right)^q+\mathbb{F}_2\right)
\\&=c+c^q+(e+1)\cdot\mathbb{F}_2,
\end{align*}
from which it follows
\begin{equation}\label{sol_eneq1}
x=c+\omega\cdot(e+1)\cdot\mathbb{F}_2+\alpha \text{ for some }
\alpha\in \mathbb{F}_q
\end{equation}
where $\omega$ is an element in $\mathbb{F}_{q^2}$ such that
$\omega+\omega^q=1$. Again, substituting $x=c+\alpha$ and
$x=c+\omega\cdot(e+1)+\alpha$ to the equation \eqref{eq1}, with
thanks to Proposition~\ref{alpha}, yields
\begin{equation}\label{eq3}
\alpha^2+\alpha=f(c)+b
\end{equation}
\begin{equation}\label{eq4}
\alpha^2+\alpha=f\left(c+\omega(e+1)\right)+b.
\end{equation}
Consequently, via Proposition~\ref{quadratic_old} we have
\begin{proposition}\label{prop2} Assume $b\in \GF{q^4}$ satisfies
$e:=\mathbf{T}^{4k}_k(b)\neq1$. Let $\omega$ be an element in
$\mathbb{F}_{q^2}$ such that $\omega+\omega^q=1$, and $c\in
\GF{q^4}$ be an element such that $c^2+(e+1)c=b^q.$ Then, the
followings hold.
\begin{enumerate}
\item $f(X)=b$ has either 0, 2 or 4 solutions in
$\GF{q^4}$.
\item $f(X)=b$ has no solution if{f} $$\mathbf{T}_k(f(c)+b)\neq0 \text{ and } \mathbf{T}_k\left(f\left(c+\omega(e+1)\right)+b\right)\neq 0.$$
\item $f(X)=b$ has 2 solutions in
$\GF{q^4}$ if{f} either of the following conditions is satisfied :
\begin{enumerate}
\item $\mathbf{T}_k(f(c)+b)=0 \text{ and } \mathbf{T}_k\left(f\left(c+\omega(e+1)\right)+b\right)\neq 0;$
\item $\mathbf{T}_k(f(c)+b)\neq0 \text{ and } \mathbf{T}_k\left(f\left(c+\omega(e+1)\right)+b\right)= 0.$
\end{enumerate}
\item $f(X)=b$ has 4 solutions if{f}
$$\mathbf{T}_k(f(c)+b)=0 \text{ and } \mathbf{T}_k\left(f\left(c+\omega(e+1)\right)+b\right)= 0.$$
\item The explicit expressions
of the solutions in $\GF{q^4}$ to $f(X)=b$ are given by
\eqref{sol_eneq1}, \eqref{eq3} and \eqref{eq4} with application of
Proposition~\ref{quadratic_old}.
\end{enumerate}\qed
\end{proposition}
\subsection{Restating the conditions on the solution number}
Though the main subject of this article to explicitly solve the
equation $f(X)=b$, is resolved by Propositions~\ref{prop1}
and~\ref{prop2} with Proposition~\ref{quadratic_old},  yet it is not
clear how these results lead to a new way of determining the
differential spectrum of $X^{q^2+q+1}$. This subsection contributes
to resolve the problem.

\begin{proposition}\label{prop3} Let $b,c\in \GF{q^4}$satisfy $$c^2+(e+1)c=b^q$$ where $e:=\mathbf{T}_k^{4k}(b)$. Then, the
followings hold.
\begin{enumerate}
\item $\mathbf{T}_k^{4k}(c)\in \{1,e\}$.
\item If either $f(c)+b$ or $f\left(c+\omega(e+1)\right)+b$ lies in
$\GF{q}$, then
$$\mathbf{T}_k^{4k}(c)=e.$$
\item If $\mathbf{T}_k^{4k}(c)=e$, then
\[\mathbf{T}_k(f(c)+b)=\mathbf{T}_k\left((1+c+c^{q^2})^{q+1}\right)+\mathbf{T}_{2k}\left(c^{q^2+1}\right)\] and
\[
\mathbf{T}_k\left(f(c+\omega(e+1))+b\right)=\mathbf{T}_{2k}\left(c^{q^2+1}\right)+1.
\]
\end{enumerate}
\end{proposition}
\begin{proof}
Applying the operator $\mathbf{T}_k^{4k}$ to both sides of
$c^2+(e+1)c=b^q$ gives
$\mathbf{T}_k^{4k}(c)^2+(\mathbf{T}_k^{4k}(b)+1)\mathbf{T}_k^{4k}(c)=\mathbf{T}_k^{4k}(b)$,
i.e.,
$\left(\mathbf{T}_k^{4k}(c)+1\right)\cdot\left(\mathbf{T}_k^{4k}(c)+\mathbf{T}_k^{4k}(b)\right)=0,$
which proves the first item.

Now, $f(c)+f(c)^q=b+b^q$ is rewritten as
\[c^{q+1}+c^{q^3+q^2}+c^{q^2+1}+c^{q^3+q}+c+c^{q^3}=b+b^q,\] or
equivalently, as
\[(c+c^q)^{q^3+q}+(c+c^q)^{q^3}=b+b^q.\]
 By raising this equality to the $q^2-$th power, we get
 \[(c+c^q)^{q^3+q}+(c+c^q)^{q}=b^{q^2}+b^{q^3}\] and then adding it to the
 original
 equality yields $\mathbf{T}_k^{4k}(c)=e$. Further, then we have $b^q=c^2+(e+1)c=c^{q^3+1}+c^{q^2+1}+c^{q+1}+c$
 from which it follows
$$b=c^{q^3+q^2}+c^{q^3+q}+c^{q^3+1}+c^{q^3}$$
and
\begin{align*}
&f(c)+b=c^{q^2+q}+c^{q^2+1}+c^{q+1}+c^{q^3+q^2}+c^{q^3+q}+c^{q^3+1}+e+1\\
&=\mathbf{T}^{4k}_k(c+c^{q+1})+\mathbf{T}^{2k}_k(c^{q^2+1})+1.
\end{align*}
Therefore, by making successive uses of Item~\ref{Tp2} of
Proposition~\ref{P1} we get
\begin{align*}
&\mathbf{T}_k(f(c)+b)=\mathbf{T}_{4k}\left(c+c^{q+1}\right)+\mathbf{T}_{2k}\left(c^{q^2+1}\right)+\mathbf{T}_k(1)\\
&=\mathbf{T}_{k}\left((c^{q^2}+c)^q+(c^{q^2}+c)+(c^{q^2}+c)^{q+1}+1\right)+\mathbf{T}_{2k}\left(c^{q^2+1}\right)\\
&=\mathbf{T}_k\left((1+c+c^{q^2})^{q+1}\right)+\mathbf{T}_{2k}\left(c^{q^2+1}\right)
\end{align*}

On the other hand, by exploiting the second equality of
Proposition~\ref{alpha} we can get
$$f\left(c+\omega(e+1)\right)=f(c)+(e+1)(c+c^{q^2})+(e+1)^2\omega^2+(e+1)\omega^q,$$
and then
$f\left(c+\omega(e+1)\right)+f\left(c+\omega(e+1)\right)^q=b+b^q$ is
rewritten as
\begin{align*}
&f(c)+(e+1)(c+c^{q^2})+(e+1)^2\omega^2+(e+1)\omega^q\\
&+\left(f(c)+(e+1)(c+c^{q^2})+(e+1)^2\omega^2+(e+1)\omega^q\right)^q=b+b^q,
\end{align*}
or equivalently, as
\[(c+c^q)^{q^3+q}+(c+c^q)^{q^3}+(e+1)(\mathbf{T}^{4k}_k(c)+e)=b+b^q.\]
By raising this equality to the $q^2-$th power, we get
 \[(c+c^q)^{q^3+q}+(c+c^q)^{q}+(e+1)(\mathbf{T}^{4k}_k(c)+e)=b^{q^2}+b^{q^3}\] and then adding it to the
 original
 equality gives $\mathbf{T}_k^{4k}(c)=e$.
 Further, then
 \begin{align*}
&\mathbf{T}_k\left(f(c+\omega(e+1))+b\right)+\mathbf{T}_k\left(f(c)+b\right)\\
&=\mathbf{T}_k\left((e+1)(c+c^{q^2})+(e+1)^2\omega^2+(e+1)\omega^q\right)\\
&=\mathbf{T}_k\left(c^{q^2}+c+(c^{q^2}+c)^2+(c^{q^2}+c)^{q+1}+(e+1)^2\omega^2+(e+1)\omega+e+1\right)\\
&\overset{\text{Prop.
\ref{P1}}}{=}(c^{q^2}+c)+(c^{q^2}+c)^q+\mathbf{T}_{k}((c^{q^2}+c)^{q+1})+e+1+\mathbf{T}_k(e+1)\\
 &=\mathbf{T}_{k}\left((c^{q^2}+c)^{q+1}+e+1\right)+1\\
 &=\mathbf{T}_k\left((1+c+c^{q^2})^{q+1}\right)+1.
\end{align*}
\qed
\end{proof}

With Proposition~\ref{prop1}, Proposition~\ref{prop2} and
Proposition~\ref{prop3}, we can state
\begin{theorem}\label{main theorem}
Let $b$ be any element in $\GF{q^4}$ and $c\in \GF{q^4}$ a solution
to
$$c^2+(\mathbf{T}_k^{4k}(b)+1)c=b^q.$$ Then, for $f(X)=X^{q^2+q+1}+(X+1)^{q^2+q+1}$, the following are true.
\begin{enumerate}
\item $f(X)=b$ has either 0, 2 or 4 solutions in $\GF{q^4}$.
\item $f(X)=b$ has no solution in $\GF{q^4}$ if{f} the following hold
true :
\begin{enumerate}
\item $\mathbf{T}_k^{4k}(b)=1$ and
$\mathbf{T}_{2k}\left(b^{q^2+1}\right)=0$, or,
\item $\mathbf{T}_k^{4k}(c)=1,$ or,
\item $\mathbf{T}_k^{4k}(c)\neq1  \text{ and } \mathbf{T}_k\left((1+c+c^{q^2})^{q+1}\right)=1 \text{ and }\mathbf{T}_{2k}\left(c^{q^2+1}\right)=0.$
\end{enumerate}
\item $f(X)=b$ has 2 solutions in $\GF{q^4}$ if{f} the following hold
true :
\begin{enumerate}
\item $\mathbf{T}_k^{4k}(b)=1$ and
$\mathbf{T}_{2k}\left(b^{q^2+1}\right)=1$, or,
\item $\mathbf{T}_k^{4k}(c)\neq 1$ and $\mathbf{T}_k\left((1+c+c^{q^2})^{q+1}\right)=0.$
\end{enumerate}
\item $f(X)=b$ has 4 solutions in $\GF{q^4}$ if{f} $$\mathbf{T}_k^{4k}(c)\neq1  \text{ and } \mathbf{T}_k\left((1+c+c^{q^2})^{q+1}\right)=\mathbf{T}_{2k}\left(c^{q^2+1}\right)=1.$$
\item  The explicit expressions
of the solutions in $\GF{q^4}$ to $f(X)=b$ are given by
\eqref{sol1x}, \eqref{sol1a}, \eqref{sol_eneq1}, \eqref{eq3} and
\eqref{eq4} with application of Proposition~\ref{quadratic_old}.
\end{enumerate}
\end{theorem}
\begin{proof} The first and fifth items were already known from
Proposition~\ref{prop1} and Proposition~\ref{prop2}.

From the first two items of Proposition~\ref{prop3}, it follows : if
$\mathbf{T}_k^{4k}(c)=1$, then neither $f(c)+b$ nor
$f\left(c+\omega(e+1)\right)+b$ lies in $\GF{q}$, and thus, by the
equivalence $\mathbf{T}_k(y)\in \GF{2}\Longleftrightarrow y\in
\GF{q}$ and Proposition~\ref{prop2}, $f(X)=b$ has no solution in
$\GF{q^4}$. Now, $\mathbf{T}_k^{4k}(c)\neq1$ which also implies
$\mathbf{T}_k^{4k}(b)\neq1$ is equivalent to
$\mathbf{T}_k^{4k}(c)=e$ via the first item of
Proposition~\ref{prop3}, and the third item of
Proposition~\ref{prop3} applied to Proposition~\ref{prop2} proves
the remaining items, combined with Proposition~\ref{prop1}.\qed
\end{proof}

To compute the differential spectrum, we further need the following
two facts.
\begin{proposition}\label{prop10} For every $i\in \{0,1\}$, it holds
\[
\#\left\{b\in \mathbb{F}_{q^4}| \mathbf{T}^{4k}_k(b)=1,
\mathbf{T}_{2k}(b^{q^2+1})=i\right\}=\frac{q^3}{2}.
\]
\end{proposition}
\begin{proof}
Define sets $$S:=\left\{b\in \mathbb{F}_{q^4}\mid
\mathbf{T}^{4k}_k(b)=1\right\}$$ and $$S_i:=\left\{b\in
\mathbb{F}_{q^4}\mid \mathbf{T}^{4k}_k(b)=1,
\mathbf{T}_{2k}(b^{q^2+1})=i\right\}.$$ Obviously $S_0\cup S_1=S$
and $\#S=q^3$ as $\mathbf{T}^{4k}_k$ is a $q^3$-to-1 mapping  from
$\GF{q^4}$ onto $\GF{q}$. Fixing any $u\in \mathbb{F}_q$ with
$\mathbf{T}_k(u)=1$, the mapping $b\in S_0\longmapsto b+u\in S_1$ is
a well-defined bijection between $S_0$ and $S_1$. Indeed, for any
$b\in S_0$, one has $\mathbf{T}^{4k}_k(b+u)=\mathbf{T}^{4k}_k(b)=1$
and
\begin{align*}
  &\mathbf{T}_{2k}\left((b+u)^{q^2+1}\right)=\mathbf{T}_{2k}\left(b^{q^2+1}+u(b+b^{q^2})+u^2\right)
 \\&=\mathbf{T}_{2k}\left(u(b+b^{q^2})+u^2\right)\overset{\text{Item \ref{Tp2} of Prop. \ref{P1}}}{=}
 \mathbf{T}_k\left(\mathbf{T}_{k}^{2k}\left(u(b+b^{q^2})+u^2\right)\right)\\
 &=\mathbf{T}_k\left(\mathbf{T}_{k}^{2k}\left(u(b+b^{q^2})\right)\right)=\mathbf{T}_k(u\mathbf{T}^{4k}_k(b))=\mathbf{T}_k(u)=1,
\end{align*}
thus $b+u\in S_1.$ Therefore
$\#S_0=\#S_1=\frac{\#S}{2}=\frac{q^3}{2}$.\qed
\end{proof}

\begin{proposition}\label{prop11} It holds
\[
\#\left\{c\in \mathbb{F}_{q^4}\mid \mathbf{T}^{4k}_k(c)\neq 1,
\mathbf{T}_k\left((1+c+c^{q^2})^{q+1}\right)=0\right\}=\frac{q^4-q^3}{2}.
\]
\end{proposition}
\begin{proof}
If $\mathbf{T}^{4k}_k(c)=1$, then
\begin{align*}
&\mathbf{T}_k\left((1+c+c^{q^2})^{q+1}\right)=\mathbf{T}_k\left((c+c^{q^2})(1+c+c^{q^2})\right)\\
&=\mathbf{T}_k\left((c+c^{q^2})+(c+c^{q^2})^2\right)\overset{\text{Item
\ref{Tp1} of Prop.
\ref{P1}}}{=}(c+c^{q^2})+(c+c^{q^2})^q=\mathbf{T}^{4k}_k(c)=1.
\end{align*}
Therefore we have in fact $$\left\{c\in \mathbb{F}_{q^4}\mid
\mathbf{T}^{4k}_k(c)\neq 1,
\mathbf{T}_k\left((1+c+c^{q^2})^{q+1}\right)=0\right\}=\left\{c\in
\mathbb{F}_{q^4}\mid
\mathbf{T}_k\left((1+c+c^{q^2})^{q+1}\right)=0\right\}.$$ To
determine its cardinality, note the following three facts :
\begin{enumerate}
\item $c\in \GF{q^4}\longmapsto
1+c+c^{q^2}=1+\mathbf{T}_{2k}^{4k}(c)\in \GF{q^2}$ is a $q^2$-to-1
mapping from $\GF{q^4}$ onto $\GF{q^2}$;
\item $d\in
\GF{q^2}^*\longmapsto d^{q+1}\in \GF{q}^*$ is a $(q+1)$-to-1 mapping
from $\GF{q^2}^*$ onto $\GF{q}^*$;
\item In $\GF{q}$ there are
exactly $\frac{q}{2}$ elements (including 0) with the absolute trace
0.
\end{enumerate} Consequently, we have
$$\#\left\{c\in
\mathbb{F}_{q^4}\mid
\mathbf{T}_k\left((1+c+c^{q^2})^{q+1}\right)=0\right\}=q^2\cdot\left(1+\left(\frac{q}{2}-1\right)(q+1)\right)
=\frac{q^4-q^3}{2}.$$\qed
\end{proof}

Finally, we have the following result which was also proved in
\cite{BL10,XY17,Fu22}.
\begin{corollary}\label{cor}
Let $N_i$ be the number of $b$'s in $\mathbb{F}_{q^4}$ such that
$$X^{q^2+q+1}+(X+1)^{q^2+q+1}=b$$ has $i$ solutions in $\GF{q^4}$.
Then
$$N_0=\frac{5q^4-q^3}{8}, N_2=\frac{q^4+q^3}{4},
N_4=\frac{q^4-q^3}{8}.$$
\end{corollary}
\begin{proof}
 For each $b\in \GF{q^4}$, there are two different $c$'s ($c_0$ and $c_0+e+1$) satisfying
$c^2+(e+1)c=b^q$. By Theorem~\ref{main theorem},
Proposition~\ref{prop10} and Proposition~\ref{prop11}, we have
$N_2=\frac{q^3}{2}+\frac{1}{2}\cdot\frac{q^4-q^3}{2}=\frac{q^4+q^3}{4}$,
and therefore $N_4=\frac{q^4-2\cdot N_2}{4}=\frac{q^4-q^3}{8},$
$N_0=q^4-N_2-N_4=\frac{5q^4-q^3}{8}.$ \qed

\end{proof}


\end{document}